\documentclass[sigconf,nonacm]{acmart}

\usepackage{booktabs} 
\usepackage{graphicx}
\usepackage{latexsym}
\usepackage{amsmath}
\usepackage{amssymb}
\usepackage{url}
\usepackage{float}
\usepackage{tikz}
\usepackage{fancyvrb}
\usepackage{algpseudocode}
\usepackage{listings}
\usepackage{xcolor}
\usepackage{color}
\usepackage{algorithm}
\usepackage{soul}
\usepackage{multirow}
\colorlet{TufteRed}{red!80!black}

\newcommand{\hide}[1]{}
\newcommand{\xhdr}[1]{\vspace{0.5mm}\noindent{\textbf{#1.}}\hspace{0.5mm}}

\newcommand{\UsizeC}{U(\lvert C \rvert)}

\usepackage[capitalize]{cleveref}

\newtheorem{theorem}{Theorem}
\newtheorem{coro}[theorem]{Corollary}

\newtheorem{lemma}[theorem]{Lemma}

\newtheorem{definition}[theorem]{Definition}

\setcopyright{rightsretained}

\usepackage[skip=0pt]{caption}
\acmDOI{}

\acmISBN{}

\acmConference[]{}{}{}
\acmYear{}
\copyrightyear{}

\acmArticle{}
\acmPrice{}

\begin{document}
\title{Planted Hitting Set Recovery in Hypergraphs}

\author{Ilya Amburg}
\orcid{0003-1632-5427}
\affiliation{%
  \institution{Cornell University}
  \city{Ithaca}
  \state{New York}
  \postcode{14850}
}
\email{ia244@cornell.edu}

\author{Jon Kleinberg}
\affiliation{%
  \institution{Cornell University}
  \city{Ithaca}
  \state{New York}
  \postcode{14850}
}
\email{kleinber@cs.cornell.edu}

\author{Austin R.~Benson}
\orcid{0001-6110-1583}
\affiliation{%
  \institution{Cornell University}
  \city{Ithaca}
  \state{New York}
  \postcode{14850}
}
\email{arb@cs.cornell.edu}

%
%


\begin{abstract}
In various application areas, networked data is collected by measuring
interactions involving some specific set of {\em core} nodes.  This results in a
network dataset containing the core nodes along with a potentially much larger
set of fringe nodes that all have at least one interaction with a core
node. 
In many settings, this type of data arises for structures that are richer than graphs, because they involve the interactions of larger sets; 
for example, the core nodes might be a set of individuals under
surveillance, where we observe the attendees of meetings involving at least one of
the core individuals.
We model such scenarios using {\em hypergraphs}, and we
study the problem of core recovery: if we observe the
hypergraph but not the labels of core and fringe nodes, can we recover the
``planted'' set of core nodes in the hypergraph?

We provide a theoretical framework for analyzing the recovery of such a
set of core nodes and use our theory to develop a practical and scalable
algorithm for core recovery. The crux of our analysis and algorithm is that the core
nodes are a \emph{hitting set} of the hypergraph, meaning that every hyperedge
has at least one node in the set of core nodes.  We demonstrate the efficacy of
our algorithm on a number of real-world datasets, outperforming competitive
baselines derived from network centrality and core-periphery measures.
\end{abstract}

\maketitle


\section{Introduction}
The data we can collect in practice is typically incomplete in several fundamental and recurring
ways, and this is particularly true for graph and network data modeling complex
systems throughout the social and biological
sciences~\cite{Laumann-1989-boundary,Kuny-1997-dark,Kossinets-2006-missing,Gile-2010-RDS,Guimer-2009-missing}.
One common type of incompleteness arises from a particular type of data
collection procedure where one wants to record all interactions involving a set
of core nodes $C$. This results in what is sometimes called a ``core-fringe''
structure~\cite{Benson-2018-planted}, where the resulting network structure
contains the core nodes along with a (potentially) much larger set of fringe
nodes $F$ that are observed in some interaction with a node in $C$. For example,
in survey data, $C$ might represent a set of respondents, and $F$ a set of
individuals with whom they interact; this is a common result in snowball and
respondent-driven
sampling~\cite{Heckathorn-2011-snowball,Gile-2010-RDS,Goel-2010-RDS}.
Another scenario arises in restrictions on surveillance.  For example, the
Enron~\cite{Klimt-2004-Enron} and
Avocado\footnote{https://catalog.ldc.upenn.edu/LDC2015T03} email datasets are
common data sources for email and network research; both datasets have a core
$C$ comprised of employees of a company and a fringe $F$ of people outside the
company that come from emails sent from or received by members of $C$.  

\begin{figure}[tb]
\centering
\scalebox{0.7}{\tikzset{every picture/.style={line width=0.75pt}} 

\begin{tikzpicture}[x=0.75pt,y=0.75pt,yscale=-1,xscale=1]

\draw  [fill={rgb, 255:red, 122; green, 188; blue, 162 }  ,fill opacity=0.21 ] (413.32,89.43) .. controls (406.82,98.36) and (388.88,96.4) .. (373.24,85.04) .. controls (357.61,73.67) and (350.19,57.22) .. (356.68,48.28) .. controls (363.18,39.35) and (381.12,41.31) .. (396.76,52.67) .. controls (412.39,64.04) and (419.81,80.49) .. (413.32,89.43) -- cycle ;
\draw  [fill={rgb, 255:red, 138; green, 158; blue, 205 }  ,fill opacity=0.21 ] (402.16,94.11) .. controls (391.66,90.69) and (387.98,73.02) .. (393.96,54.64) .. controls (399.93,36.26) and (413.29,24.12) .. (423.79,27.53) .. controls (434.3,30.95) and (437.97,48.62) .. (432,67) .. controls (426.03,85.38) and (412.67,97.52) .. (402.16,94.11) -- cycle ;
\draw  [fill={rgb, 255:red, 197; green, 180; blue, 180 }  ,fill opacity=0.2 ] (395,84) .. controls (395,72.95) and (410.67,64) .. (430,64) .. controls (449.33,64) and (465,72.95) .. (465,84) .. controls (465,95.05) and (449.33,104) .. (430,104) .. controls (410.67,104) and (395,95.05) .. (395,84) -- cycle ;
\draw  [fill={rgb, 255:red, 148; green, 136; blue, 112 }  ,fill opacity=0.21 ] (413.32,77.43) .. controls (419.81,86.36) and (412.39,102.82) .. (396.76,114.18) .. controls (381.12,125.54) and (363.18,127.51) .. (356.68,118.57) .. controls (350.19,109.64) and (357.61,93.18) .. (373.24,81.82) .. controls (388.88,70.46) and (406.82,68.49) .. (413.32,77.43) -- cycle ;
\draw  [fill={rgb, 255:red, 139; green, 87; blue, 42 }  ,fill opacity=0.21 ] (401.18,73.71) .. controls (411.69,70.3) and (425.05,82.44) .. (431.02,100.82) .. controls (436.99,119.2) and (433.32,136.87) .. (422.82,140.29) .. controls (412.31,143.7) and (398.95,131.56) .. (392.98,113.18) .. controls (387.01,94.8) and (390.68,77.13) .. (401.18,73.71) -- cycle ;
\draw  [color={rgb, 255:red, 0; green, 0; blue, 0 }  ][line width=4] [line join = round][line cap = round] (398.3,102.2) .. controls (398.3,102.2) and (398.3,102.2) .. (398.3,102.2) ;
\draw  [color={rgb, 255:red, 0; green, 0; blue, 0 }  ][line width=4] [line join = round][line cap = round] (405.3,84.2) .. controls (405.3,84.2) and (405.3,84.2) .. (405.3,84.2) ;
\draw  [color={rgb, 255:red, 0; green, 0; blue, 0 }  ][line width=4] [line join = round][line cap = round] (421.3,96.2) .. controls (421.3,96.2) and (421.3,96.2) .. (421.3,96.2) ;
\draw  [color={rgb, 255:red, 0; green, 0; blue, 0 }  ][line width=4] [line join = round][line cap = round] (399.3,64.2) .. controls (399.3,64.2) and (399.3,64.2) .. (399.3,64.2) ;
\draw  [color={rgb, 255:red, 0; green, 0; blue, 0 }  ][line width=4] [line join = round][line cap = round] (383.3,83.2) .. controls (383.3,83.2) and (383.3,83.2) .. (383.3,83.2) ;
\draw  [color={rgb, 255:red, 0; green, 0; blue, 0 }  ][line width=4] [line join = round][line cap = round] (422.3,72.2) .. controls (422.3,72.2) and (422.3,72.2) .. (422.3,72.2) ;
\draw  [fill={rgb, 255:red, 122; green, 188; blue, 162 }  ,fill opacity=0.21 ] (413.32,206.43) .. controls (406.82,215.36) and (388.88,213.4) .. (373.24,202.04) .. controls (357.61,190.67) and (350.19,174.22) .. (356.68,165.28) .. controls (363.18,156.35) and (381.12,158.31) .. (396.76,169.67) .. controls (412.39,181.04) and (419.81,197.49) .. (413.32,206.43) -- cycle ;
\draw  [fill={rgb, 255:red, 138; green, 158; blue, 205 }  ,fill opacity=0.21 ] (402.16,211.11) .. controls (391.66,207.69) and (387.98,190.02) .. (393.96,171.64) .. controls (399.93,153.26) and (413.29,141.12) .. (423.79,144.53) .. controls (434.3,147.95) and (437.97,165.62) .. (432,184) .. controls (426.03,202.38) and (412.67,214.52) .. (402.16,211.11) -- cycle ;
\draw  [fill={rgb, 255:red, 197; green, 180; blue, 180 }  ,fill opacity=0.2 ] (395,201) .. controls (395,189.95) and (410.67,181) .. (430,181) .. controls (449.33,181) and (465,189.95) .. (465,201) .. controls (465,212.05) and (449.33,221) .. (430,221) .. controls (410.67,221) and (395,212.05) .. (395,201) -- cycle ;
\draw  [fill={rgb, 255:red, 148; green, 136; blue, 112 }  ,fill opacity=0.21 ] (413.32,194.43) .. controls (419.81,203.36) and (412.39,219.82) .. (396.76,231.18) .. controls (381.12,242.54) and (363.18,244.51) .. (356.68,235.57) .. controls (350.19,226.64) and (357.61,210.18) .. (373.24,198.82) .. controls (388.88,187.46) and (406.82,185.49) .. (413.32,194.43) -- cycle ;
\draw  [fill={rgb, 255:red, 139; green, 87; blue, 42 }  ,fill opacity=0.21 ] (401.18,190.71) .. controls (411.69,187.3) and (425.05,199.44) .. (431.02,217.82) .. controls (436.99,236.2) and (433.32,253.87) .. (422.82,257.29) .. controls (412.31,260.7) and (398.95,248.56) .. (392.98,230.18) .. controls (387.01,211.8) and (390.68,194.13) .. (401.18,190.71) -- cycle ;
\draw  [color={rgb, 255:red, 0; green, 0; blue, 0 }  ][line width=4] [line join = round][line cap = round] (405.3,201.2) .. controls (405.3,201.2) and (405.3,201.2) .. (405.3,201.2) ;

\end{tikzpicture}}
\scalebox{0.9}{\input{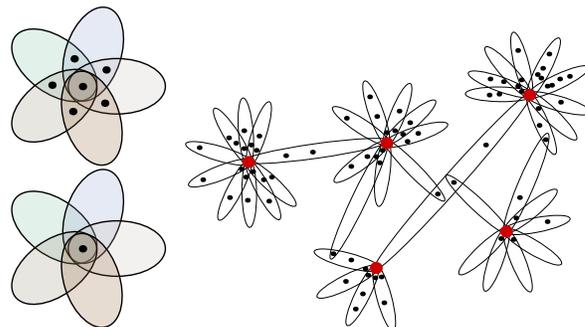}}
\caption{Hypergraphs and planted hitting sets.
{\textbf (Left)} A 3-uniform hypergraph on the top and a minimal hitting set on the bottom
(with hyperedges superimposed to facilitate comparison to the hypergraph).
{\bf (Right)} A planted hitting set (red nodes) in a hypergraph, constituting the \emph{core} nodes. 
Hyperedges are illustrated by ellipses encompassing nodes, which include both core and non-core (\emph{fringe}) nodes.
We study how well we can recover the core (red) nodes if they are not identified.
}
\label{fig:hypergraph_hs}
\end{figure}

One can imagine similar scenarios in a variety of settings. In intelligence
data, one might record all of the attendees at meetings that involve at least
one of a set of individuals under surveillance. In a similar manner,
telecommunications providers can observe group text messages where only some of
the members of the group are subscribers. Furthermore, measurements of Internet
structure from service providers at the IP-layer from
individual service providers provide only a partial
measurement~\cite{Tsiatas-2013-spectral}, and large-scale Web crawls do not
encapsulate the entire network structure of
hyperlinks~\cite{Boldi-2004-UbiCrawler,Boldi-2014-BUbiNG}.

While a dataset may be equipped with labels of which nodes are members of the
core and which are members of the fringe, there are several situations in which
such labels are unavailable. For example, consider the email datasets described
above. Such data could be released by a hacker from outside the organization, or a leaker from inside it, who has collected
email from a set of core accounts but not released the identity of the
accounts~\cite{Hier-2009-surveillance}. Similarly, in the case of intelligence
data that records attendees of meetings, the data could be released with the
identities of the individuals under surveillance redacted. In other cases, the
metadata may be lost simply to issues of data maintenance, and such concerns are
central to the research areas of data provenance and preservation
\cite{Buneman-2001-provenance,Lynch-2008-data,Simmhan-2005-survey,Tan-2004-provenance}.

The question of recovering the identity of the core nodes in a dataset with core-fringe structure is therefore a fundamental question arising in different applications, and for multiple different reasons.

\xhdr{Core-Fringe Recovery in Hypergraphs}
Here, we study the problem of identifying the set of core nodes, when the labels
of which nodes are core and which are fringe are unavailable and the data is
represented by a {\em hypergraph}---that is, a set of nodes along with a set of {\em hyperedges}
that each connect multiple nodes. \Cref{fig:hypergraph_hs}
illustrates this model. The hypergraph model is apt for many of the scenarios
described above, including emails, meeting attendees, and group messages. While
the hypergraph model is richer than a more standard graph model, the
higher-order structure also presents challenges. If a hyperedge connects a
(latent) core node and several fringe nodes, a core recovery algorithm will have
to consider all nodes in the hyperedge as possible candidates for the core.

We first provide theoretical results that motivate the development
of an algorithm for recovery of the core. The key structural property of our
analysis is that the core nodes constitute a \emph{hitting set}, meaning every
hyperedge in the hypergraph has at least one node in the core
(\cref{fig:hypergraph_hs}). This property comes from how the data is
measured---every included hyperedge comes from an interaction
of a member of the core with some other set of nodes.
We thus think of the core as a ``planted'' hitting set that we must find in the data.
Our first theoretical results 
rely on two assumptions (both relaxable, to some extent): the planted hitting
set corresponding to the core is
(i) minimal, that is, no node could be removed to make the hitting set smaller; and
(ii) relatively small, i.e., bounded by a constant $k$.
We analyze how large the \emph{union of all minimal hitting sets}, denoted $U(k)$,
can be.  We build upon results in fixed-parameter tractability~\cite{Damaschke-2006-parameterized,Damaschke-2009-unions}
to show that $\lvert U(k) \rvert$ is $\Theta(k^r)$
in the worst case, where $r$ is the size of the largest hyperedge; importantly,
this bound is independent of the hypergraph's size, and $U(k)$ is
guaranteed to contain our planted hitting set if it meets the modeling
assumptions.

Furthermore, we prove that a classical greedy approximation algorithm for set
cover relates to partial recovery of the planted hitting set. In a hypergraph
where $r$ is the size of the largest hyperedge, we show that the output of the
algorithm must overlap with the planted hitting set by at least a 
$O(1 / r)$ fraction of the nodes, provided that the hitting set size is within a constant
factor of the minimum hitting set size.

Combining these two main results leads to our algorithm, which we call the union
of minimal hitting sets (UMHS). This algorithm is practical and
scalable, using only simple subroutines and having computational complexity
roughly linear in the size of the data. The idea is to run the greedy algorithm
multiple times with random initialization to find several minimal hitting sets
and then to simply take the union of these hitting sets. The first part of our
theory says that the output should not grow too large, and the second part of
our theory says the output should also overlap the planted hitting set.  We show
that our method consistently outperforms a number of baselines derived from
network centrality and core-periphery structure for hypergraphs on
several real-world datasets.

Our approach to the planted hitting set problem is purely combinatorial.
This contrasts with typical planted problems such as the stochastic block model 
and planted clique recovery, which are based on probabilistic structure.
Thus, the methods we use are fundamentally
different than what is common for these types of problems. Moreover, we show
that if random block-model-type structure exists for our problem, then
the problem becomes substantially easier.


\section{Problem setup and theoretical results for core recovery}

We start with some formal definitions to facilitate our presentation of the
theoretical results. A hypergraph $G$ consists of a set of nodes $V$ and a set
of hyperedges $E$, each a subset of $V$ of size at least two. The \emph{rank} $r$ of a hypergraph
is the maximum size of any hyperedge, i.e., $r = \max_{e \in E} \lvert e \rvert$. A hypergraph
is called \emph{$r$-uniform} if all hyperedges have cardinality equal to $r$.

In our setup, there is some unidentified subset $C \subseteq V$ that is
designated as the set of ``core nodes.'' Our goal is to find $C$, knowing that
$C$ is a hitting set of $G$ (i.e., for every $e \in E$, there is some $v \in e \cap C$).
We say that $C$ is \emph{planted} since we do not know its identity.
Absent any additional information, we are helpless---it could be that $C =
V$. However, we can do better by assuming two properties of $C$ (both
relaxable):
(i) $C$ is a minimal hitting set in $G$ and (ii)
$\lvert C \rvert \le k$.
Under these constraints, it is certainly true that $C$ is contained in the union
of all minimal hitting sets, which is a principle object of our analysis:
\begin{definition}\label{def:Uk}
Given a hypergraph $G$, $U(k)$ is the union of all minimal hitting sets of $G$
of size at most $k$.
\end{definition}

We next show bounds on $\lvert U(k) \rvert$ that are independent of the size of
the graph. This says that $C$ is contained in a relatively small set if it
satisfies the conditions. Furthermore, we can find $U(k)$ in time exponential in
$k$ but polynomial in the size of the graph; however, we develop more practical
algorithms in \cref{sec:UMHS}. After, we relax the assumptions on minimality and get
results on \emph{partial recovery}, i.e., algorithms that are guaranteed to find a
part of $C$. Finally, we show how these results improve under a random hypergraph model.

\subsection{Minimal hitting sets}\label{sec:minimal}

First let us suppose that we are given $G$ and asked to find $C$, where we know
that $C$ is minimal and $\lvert C \rvert \leq k$. We ask the following question: is
it possible to find a small set $H\subset G$, whose size is independent of $G$,
such that $C \subset H$? In this section, we answer the question in the
affirmative.

We at least know that the union of \emph{all} minimal hitting sets of size at
most $k$, $U(k)$ (\cref{def:Uk}), \emph{must} contain $C$. But how large can
this set be? The following result says that $U(k)$ actually can't be that large.
\begin{lemma}[\cite{Damaschke-2009-unions,Damaschke-2006-parameterized}]\label{lem:UMHS}
  In a hypergraph of rank $r$, $U(k)$ has size $\Theta(k^r)$ in the worst case.
\end{lemma}

Damaschke and Molokov established the upper bound
$\lvert U(k) \rvert = O(k^r)$
using sophisticated techniques from parameterized complexity theory~\cite{Damaschke-2009-unions}.
Here, we prove the same upper bound with a simpler, self-contained
proof that uses the celebrated sunflower lemma of Erd{\"o}s and Rado~\cite{Erdos-1960-intersection}.
The lower bound $\lvert U(k) \rvert = \Omega(k^r)$ was stated by Damashke~\cite{Damaschke-2006-parameterized}
with a brief proof sketch. We provide a full proof 
for completeness.
A consequence of these results is the following:
\begin{theorem}\label{thm:UMHS}
  Let $C$ be a planted minimal hitting set with $\lvert C \rvert \le k$
  in a hypergraph of rank $r$.
  Then we can find a set $D$ of size $O(k^r)$ that is guaranteed
  to contain $C$.
\end{theorem}

We begin by providing the self-contained proof of the upper bound in
\cref{thm:UMHS}. Namely, let $G$ be a hypergraph of rank $r$, and let $U(k)$ be
the union of all minimal hitting sets in $G$ of size at most $k$. We will find a
function $g(r,k)=O(k^r)$ so that $\lvert U(k) \rvert \leq g(r, k)$.
To find such a function, we use a combinatorial structure called a
sunflower~\cite{Bollobas-1986-combinatorics}. A collection of $r$ sets
$\{X_1,X_2, . . . ., X_r\}$ is an $r$-sunflower if for all pairs of distinct
sets $X_i$, $X_j$, the pairwise intersection $X_i\cap X_j$ is equal to the
mutual intersection $\bigcap_{i=1}^{r}X_i$.  The term comes from the fact that
we can think of the $X_i$ as the ``petals" of the sunflower: they share a
``center" $\bigcap_{i=1}^{r}X_i$ and are otherwise pairwise disjoint from each
other. A famous theorem of Erd{\"o}s and Rado asserts that every sufficiently
large collection of sets must contain a large sunflower.

\begin{lemma}[Sunflower Lemma~\cite{Erdos-1960-intersection}]\label{lem:sunflower}
  Define the function $\sigma(r, k) = r!(k - 1)^r$. Any
  collection of more than $\sigma(r, k)$ sets, each of size at most $r$, must
  contain a $k$-sunflower.
\end{lemma}

We use this result to find the desired function $g(r,k)=O(k^r)$.
Subsequently, we establish an asymptotically matching lower bound.

\begin{lemma}
\label{lem:upper}
Let $G$ be a hypergraph of rank $r$. Then $\lvert U(k) \rvert = O(k^r)$.
\end{lemma}

\begin{proof}
We assume that $G$ has at least one hitting set of size at most $k$, as
otherwise the statement of \cref{lem:upper} holds with $U(k)$ equal to the
empty set.  We modify $G$ according to the following iterative procedure, which
operates in discrete phases.  We start by defining $H_0 = G$, and produce a
succession of hypergraphs $H_1, H_2, ...$, with $H_t$ the hypergraph at the end
of phase $t$. For each hypergraph $H_t$, let $C(H_t)$ denote the set of all
hitting sets of size at most $k$ in $H_t$, and let $\mathcal{C}^*(H_t)$
denote the set of all minimal hitting sets of size at most $k$ in $H_t$. We will
establish by induction that each hypergraph $H_t$ has at least one
hitting set of size at most $k$.

At the start of phase $t$, we ask whether $H_{t-1}$ has more than $\sigma(r, k + 1)$ 
hyperedges, where $\sigma$ is the function from \cref{lem:sunflower}. If it doesn't, we declare the procedure to be over. 
If it does, then we find a $(k + 1)$-sunflower in the set of hyperedges, consisting
of hyperedges $X_{t,1}, X_{t,2},\ldots, X_{t,k+1}$.

Now, the core of the sunflower found in phase $t$, which we will denote by
$Y_t = \bigcap_{i=1}^{k+1}X_{t,i}$ must be non-empty, for if it were empty, then
$H_{t-1}$ would contain $k + 1$ pairwise disjoint hyperedges, and hence could
not have a hitting of size at most $k$. 
We define $H_{t}$ as follows: starting from $H_{t-1}$, we remove
the hyperedges $X_{t,1}, X_{t,2}, \ldots, X_{t,k+1}$ and add the
hyperedge $Y_t$.
We observe the following properties.
\begin{enumerate}
\item Every hitting set in $H_{t-1}$ of size at most $k$ must intersect $Y_t$.
\item $C(H_{t-1}) = C(H_t)$, and hence $\mathcal{C}^*(H_{t-1})=\mathcal{C}^*(H_t)$.
\end{enumerate}
To see Property (1), 
By the definition of a sunflower, the sets $X_{t,i}\setminus Y_t$ are
pairwise disjoint for $i = 1, 2,\ldots, k+ 1$. Thus, if $C$ is a set of size
at most $k$ that doesn't contain any nodes of $Y_t$, then it must be disjoint
from at least one of the sets $X_{t,i}$, and hence it can't be a hitting set for
$H_{t-1}$.
The definition of $H_t$ is as follows: starting from $H_{t-1}$, we remove the hyperedges
$X_{t,1}, X_{t,2},\ldots, X_{t,k+1}$ and add the hyperedge $Y_t$. Since
$H_{t-1}$ and $H_{t}$ are hypergraphs on the same node sets, their respective
collections of hitting sets $C(H_{t-1})$ and $C(H_{t})$ are also over the same
sets of elements. 

Now, let $C\in \mathcal{C}(H_{t-1})$. Then by Property (1), it must intersect $Y_t$. Since
$Y_t$ is the only hyperedge in $H_t$ but not $H_{t-1}$, it follows that $C$ is a
hitting set for $H_t$, and hence $C\in \mathcal{C}(H_t)$. Conversely, let
$C'\in \mathcal{C}(H_t)$. Since $Y_t\subset X_{t,i}$ for $i = 1, 2,\ldots, k+1$,
it must be that $C_0$ intersects each $X_{t,i}$. Since
$X_{t,1}, X_{t,2},\ldots, X_{t,k+1}$ are the only hyperedges in $H_{t-1}$ but not $H_t$, it
follows that $C\in \mathcal{C}(H_{t-1})$.  Thus,
$\mathcal{C}(H_{t-1}) \subset \mathcal{C}(H_t) \subset \mathcal{C}(H_{t-1})$, and so they are equal as sets.

We apply the same process in each phase, producing $H_1$ from $G = H_0$, then
producing $H_2$ from $H_1$ and so forth. By induction using Property (2), since
$H_{t-1}$ contains at least one hitting set of size at most $k$, so does $H_t$,
as required. Since the procedure reduces the number of hyperedges in each phase,
it must eventually terminate, in some phase $t^*$. We write $H^* = H_{t^*}$ and
by applying Property (2) transitively,
\begin{equation}\label{eq:upper_3}
\mathcal{C}^*(H)=\mathcal{C}^*(H^*).
\end{equation}
We say that a node is isolated if it does not belong to any hyperedge. Let $T^*$
be the set of non-isolated nodes of $H^*$. Since the procedure stopped when
faced with $H^*$, $H^*$ had at most {$\sigma(r, k + 1)$}
hyperedges. Since each hyperedge has most $r$ elements, we have
\begin{equation}\label{eq:upper_4}
\lvert T^* \rvert \leq r\sigma(r, k + 1) = r \cdot r!k^r,
\end{equation}
and hence $\lvert T^* \rvert$ is bounded by $O(k^r)$.
Finally, recall that $U(k)$ is the union of all minimal hitting sets in $G$ of
size at most $k$. We can write $U(k)=\bigcup_{C\in\mathcal{C}^*(H)}C$, and by
\cref{eq:upper_3}, $U(k)=\bigcup_{C\in\mathcal{C}^*(H^*)}C$. No
minimal hitting set $C$ in $H^*$ can contain an isolated node $v$,
since then $C - \{v\}$ would also be a hitting set in $H^*$. From this,
$U(k) \subset T^*$, and $\lvert U(k) \rvert \le \lvert T^* \rvert = O(k^r)$
by \cref{eq:upper_4}.
\end{proof}

We now show asymptotic tightness for constant $r$, by constructing a family
of hypergraphs for which $\lvert U(k) \rvert = \Omega(k^r)$.
The following lemma combined with \cref{lem:upper} proves \cref{thm:UMHS}.

\begin{figure}[tb]
\centering
\scalebox{1.0}{\input{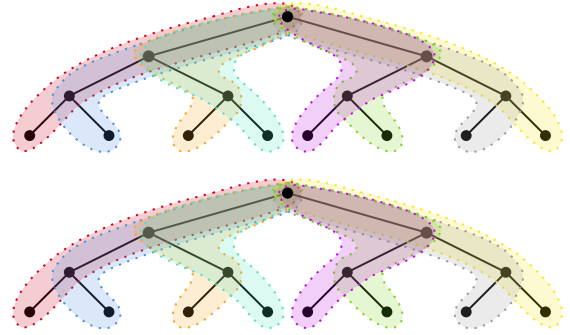}}
\caption{The hypergraph $T_2$ ($r=3$) used in the proof of \cref{lem:lower}
  with hyperedges identified by dotted shapes.}
\label{fig:trees}
\end{figure}

\begin{lemma}
\label{lem:lower}
For each constant $r \geq 2$, there exists an infinite family of rank-$r$
hypergraphs $G$ and parameters $k$, with both the number of nodes in $G$ and $k$
going to infinity, for which $\lvert U(k) \rvert = \Omega(k^r)$.
\end{lemma}

\begin{proof}
Here we present the full proof of \cref{lem:lower}. This lemma was
stated without a complete proof by Damaschke~\cite[Proposition 9]{Damaschke-2006-parameterized}.

For a parameter $b$, consider a graph $\Gamma_b$ that consists of $b$ disjoint
complete $b$-ary trees of depth $r$. We will refer to the roots of the $b$ trees
in $\Gamma_b$ as the root nodes of $\Gamma_b$, and to the leaves of the trees in
$\Gamma_b$ as leaf nodes. Let $T_b$ be the hypergraph on the node set of
$\Gamma_b$ whose hyperedges consist of all root-to-leaf paths in $\Gamma_b$,
from the root of one of the $b$ trees in $\Gamma_b$ to a leaf of the
corresponding tree (the example of $T_2$ for $r=3$ is in \cref{fig:trees}). All
edges in $G$ have size $r$.  Let $\lambda$ be any function that maps nodes of
$\Gamma_b$ to numbers in $\{1, 2, . . . ., b\}$ with the following properties:
(i) $\lambda$ is a bijection from the roots of $\Gamma_b$ to $\{1, 2, \ldots, b\}$;
and (ii) for each non-leaf node $v$ in $\Gamma_b$, $\lambda$ is a
bijection from the children of $v$ to $\{1, 2, . . . ., b\}$. For each node $v$,
we will call $\lambda(v)$ its label (with respect to $\lambda$), and we will say
that $\lambda$ is a consistent labeling if it satisfies (i) and (ii).  Let
$v^{*}_{\lambda}$ be unique leaf node for which all the nodes on the
root-to-leaf path have the label $b$.

For a consistent labeling $\lambda$, let $C_\lambda$ be the set of all nodes $v$
satisfying two properties: (i) $\lambda(u) = b$ for all nodes $u$ on the unique
path from a root of $\Gamma_b$ to $v$ (other than $v$ itself), and (ii)
$\lambda(v) \neq b$. Let $C^{*}_{\lambda} = C_\lambda\cap\{v^{*}_{\lambda}\}$.
We observe the following two facts.
\begin{enumerate}
\item For every consistent labeling $\lambda$, the set $C^{*}_{\lambda}$ is a
  minimal hitting set for $T_b$.
\item Every node v of $T_b$ belongs to at least one set of the form
  $C^{*}_{\lambda}$ for some consistent labeling $\lambda$.
\end{enumerate}
To prove (1), we claim that $C^{*}_{\lambda}$ is a hitting set.
Consider any hyperedge $X$ of $T_b$, consisting of a root-to-leaf path $P$ in
$\Gamma_b$. If all nodes on $P$ have label $b$, then the leaf node of $P$ is
$v^{*}_{\lambda}$, and hence $X$ intersects $C^{*}_{\lambda}$. Otherwise,
consider the first node $v$ on the path $P$ that has a label unequal to
$b$. Thus, $v \in  C^{*}_{\lambda}$, and again $X$ intersects
$C^{*}_{\lambda}$ and $C^{*}_{\lambda}$ is a hitting set. We now argue that it is minimal.
If we delete $v^{*}_{\lambda}$ from $C^{*}_{\lambda}$, then
$C^{*}_{\lambda}\setminus \{v^{*}_{\lambda}\}$ would be disjoint from the
hyperedge $X$ consisting of the root-to-leaf path to $v^{*}_{\lambda}$. If we
delete some other $v^{*}_{\lambda}$, then a hyperedge $X$ that passes through
$v$ does not pass through any other node of $C^{*}_{\lambda}$ on its way from
root to leaf, and so $C^{*}_{\lambda}\setminus \{v\}$ would be disjoint from
$X$.

To prove (2), we simply choose any consistent labeling $\lambda$ such that
$\lambda(u) = b$ for all nodes $u$ on the unique path from a root of $\Gamma_b$
to $v$ (other than $v$ itself), and $\lambda(v) \neq b$.

Now, we consider the size of a set of the form $C^{*}_{\lambda}$. Since
$C^{*}_{\lambda}$ contains $b - 1$ nodes from each of the first $r - 1$ levels
of the trees in $\Gamma_b$, and $b$ from the lowest level, we have
\[
\lvert C^{*}_{\lambda} \rvert = (r - 1)(b - 1) + b.
\]

We define $k = (r - 1)(b - 1) + b$. Since each set $C^{*}_{\lambda}$ is a
minimal hitting set of size $k$, from (2) we see that the union $U(k)$ of all
minimal hitting sets of size $k$ is the entire node set of $T_b$.

The number of nodes in $T_b$ is greater than $b^r$, and hence
$\lvert U(k) \rvert > b^r$.
Since $k \leq br$, we have that $\lvert U(k) \rvert > b^r\geq (k/r)^r = r^{-r}k^r$.
\end{proof}

\subsection{Non-minimal hitting sets}
\label{sec:nonminimal}

Here we present results on which types of nodes in a planted hitting
set $C$ must be contained in $U(\lvert C \rvert)$ when the hitting set is itself
not minimal. This matters in practice, when $C$ may not
be minimal.

\xhdr{Classes of nodes that must be in $\UsizeC$}
Our first result is that a node in $C$ must appear in $\UsizeC$
if all the other nodes in a hyperedge containing $u$ are not in $C$.
\begin{lemma}\label{lem:outside}
  Let $G$ be a rank-$r$ hypergraph with planted hitting setting $C$.  If $u \in C$
  and there exists a hyperedge $(u,v_1,....,v_i)$ with $i \leq r-1$ with
  $v_1,\ldots,v_i \notin C$, then $u \in \UsizeC$.
\end{lemma}
\begin{proof}
We perform a pruning of $C$ as follows. Check whether there exists a node
$w$ with $C \setminus \{w\}$ still being a hitting set of smaller
cardinality. If this is the case, update $C$ to $C \setminus \{w\}$ and continue
pruning. Upon termination, the remaining set is a minimal hitting set
$C' \subseteq C$. Clearly, $\lvert C' \rvert \leq \lvert C \rvert$ so that $C'\subset \UsizeC$.
However, $u \in U(\lvert C \rvert)$ since $(u,v_1,\ldots,v_i)$ is a hyperedge where $v_1,\ldots,v_i\notin C$.
\end{proof}

The above lemma provides a class of nodes in a planted hitting set which are
guaranteed to be in the union of minimal hitting sets of size no more than that
of the planted hitting set itself. Next, we show that nodes adjacent to other
nodes deeply integrated into $C$ must also be in $\UsizeC$.  Formally, define a
node $v$ to be in the \emph{interior} of a hitting set $C$ if all hyperedges
containing $v$ are comprised entirely of nodes in $C$.  
Our next result is that a node in $C$ must appear in $\UsizeC$
if all the other nodes in a hyperedge containing $u$ are in the
interior of $C$.
\begin{lemma}\label{lem:interior}
  Let $G$ be a rank-$r$ hypergraph with planted hitting set $C$.
  If $u \in C$ and there exists a hyperedge $(u,v_1,\ldots,v_i)$ where $v_1,\ldots,v_i$ are in
  the interior of $C$, then $u \in U(\lvert C \rvert)$.
\end{lemma}
\begin{proof}
Since $v_1,\ldots,v_i$ are in the interior of $C$,
$C_0=C\setminus\{v_1,\ldots,v_i\}$ is a hitting set as well. Now perform pruning
on $C_0$ as in the previous proof, the output of which is $C'\subset C_0$, and
since clearly $|C'|\leq \lvert C \rvert$, we must have that $C'\subset \UsizeC$.
Clearly $u$ was not deleted during pruning since $(u,v_1,\ldots,v_i)$ is a
hyperedge with $v_1,\ldots,v_i\notin C$, so that $u\in C'\subset \UsizeC$.
\end{proof}

\xhdr{Greedy matchings intersect the planted hitting set}
One way of finding hitting sets is through greedy algorithms for set cover.  The
classical matching algorithm is simple (\cref{alg:greedy}): loop through each
hyperedge and if no vertex in the current hyperedge is in the current cover $S$, add all to
the $S$. The greedy algorithm produces a set that is both a hitting set
and a maximal matching. To see that it is a hitting set, suppose that all nodes
in a hyperedge $h$ were not added. Then $h$ would have been added to $S$ at the
time the algorithm processed it. The output is a maximal matching because if
we could append another hyperedge $h$ to $S$, then $h$ would have already
been added to $S$ when the algorithm processed it.

\begin{algorithm}[tb]
\caption{Maximal matching $r$-approximation to the minimum hitting set that intersects the core.}
\label{alg:greedy}
\begin{algorithmic}[1]
\State \textbf{Input}:
Hypergraph $G=(V,E)$ of rank $r$\\
\textbf{Output}: 
Set cover $S$ with $\lvert S \rvert \leq rk^*$.
\State $S \gets \emptyset$
\For {$(u_1,....,u_i)\in E$}
\If {$u_1,....,u_i\notin S$}
\State $S \gets S \cup \{u_1,....,u_i\}$
\EndIf
\EndFor
\State \textbf{return} $S$
\end{algorithmic}
\end{algorithm}

Let $k^*$ be the \emph{minimum} hitting set size (the smallest size of all minimal hitting sets). The output of \cref{alg:greedy} is an $r$-approximation if the
hypergraph $G$ is an $r$-uniform hypergraph: in the worst case, any hitting set
contains at least one of $r$ nodes from each hyperedge, and we therefore have
$k^* \geq \lvert S \rvert / r$.
It turns out that this greedy algorithm must also partially overlap with the planted hitting set.
We formalize this in the following lemma.
\begin{lemma}
\label{lem:overlap}
Let $\lvert B \rvert \le bk^*$ for some hitting set $B$.
If the input to \cref{alg:greedy} is an rank-$r$ hypergraph, then the output $S$
satisfies $\lvert B \cap S \rvert / \lvert B \rvert \ge \frac{1}{rb}$.
\end{lemma}
\begin{proof}
Let $r_1$ be the number of vertices in the hyperedge containing the least
vertices. The set $S$ contains within itself all vertices from $h$ hyperedges where $h$
satisfies $h \leq \lvert S \rvert / r_1 \leq rk^* / r_1$. Since $B$ is a hitting set, it must
contain at least one vertex from each of the $h$ hyperedges in $S$. From this,
we obtain $\lvert S \cap B \rvert \geq h \geq k^* / r_1 \geq \frac{1}{br}\lvert B \rvert$.
\end{proof}
\Cref{lem:overlap} guarantees that the greedy algorithm will output a result that overlaps any
hitting set, including the planted one. Thus, we immediately get the following corollary.

\begin{coro}
\label{cor:plantedintersection}
If the planted hitting set $C$ has size $\lvert C \rvert\leq ck^*$ then the
output of Algorithm 1 is a set $S$ with
$\lvert S\cap C \rvert / \lvert C \rvert \geq \frac{1}{cr}$,
i.e., $S$ intersects at least a fraction of $\frac{1}{cr}$ of $C$.
\end{coro}

We can process the hyperedges in any order within \cref{alg:greedy}, and this
will be important for our practical recovery algorithm that we develop later.
Another immediate corollary of \cref{lem:overlap} says that regardless of the
processing, the outputs must be fairly similar.
\begin{coro}\label{cor:overlap}
  Any two outputs $S_1$ and $S_2$ of \cref{alg:greedy} satisfy
  $\lvert S_1 \cap S_2 \rvert \geq \frac{1}{r^2}\max(\lvert S_1 \rvert, \lvert S_2 \rvert)$.
\end{coro}

Thus far, we have made no assumptions on the structure of the hypergraph.
In the next section, we show that assuming a stochastic block model type
hypergraph structure gives substantial improvements in recovery.

\subsection{Recovery in a random hypergraph model}
Imposing structure on $C$ and $G$ will enable us to obtain more information. In
particular, we will show how we can obtain tighter bounds under certain random
graph models for 3-uniform hypergraphs.  Let us assume that the hyperedges are
generated according to a stochastic block model (SBM) for
hypergraphs~\cite{Debarghya-2014-Consistency}. One block will be the core $C$
and the other the fringe nodes $F$. We assume that there is zero
probability of a hyperedge containing only nodes in $F$. Explicitly, we will
state that the probability of a hyperedge between nodes in $C$ is $p$ while the
probability of a hyperedge containing at least one node in $C$ and at least one node in $F$ is $q$.

We first provide a lemma on the independence number of hypergraphs drawn
from the hypergraph SBM. This will help us control the size of $\lvert \UsizeC \rvert$.
\begin{lemma}\label{lem:independence}
  Let $\alpha(G)$ be the independence number of an $r$-uniform hypergraph $G$
  drawn from the SBM on $n$ vertices. Then
  \[
  \textnormal{Pr}(\alpha(G)<k)\geq 1-n^{-\frac{1}{2}\left(\frac{3r!\ln n}{2p}+(r-1)\right)}, \; k=\frac{3r!\ln n}{2p}+(r-1)
  \]
\end{lemma}
\begin{proof}
We follow a proof technique common in the combinatorics literature~\cite{Kostochka-2014-independent}.
Let $z = \binom{\lvert G \rvert}{k}$,
$S_1,\ldots,S_{z}$ the size-$k$ subsets of vertices of $G$,
and $X_i$ an indicator random variable for $S_i$ being an independent set of $G$,
so $X_i = 1$ only happens when all $\binom{k}{r}$ possible hyperedges among 
the vertices of $S_i$ do not appear in the hyperedge set of $G$.
Note that this happens with probability
$(1-p)^{\binom{k}{r}}=(1-p)^{\frac{k\cdot\cdot\cdot\cdot(k-r+1)}{r!}}\leq(1-p)^{\frac{k(k-r+1)}{r!}}$. Thus,
\begin{eqnarray*}
  \textstyle \textnormal{Pr}(\alpha(G)\geq k)
  &=& \textstyle
  \text{Pr}(\sum_iE(X_i) \geq 1) 
    \leq \sum_iE(X_i)
    = \binom{n}{k}(1-p)^{\binom{k}{r}}\\
  &\leq& \textstyle
    \binom{n}{k}(1-p)^{\frac{k(k-r+1)}{r!}}
\leq n^k\left((1-p)^{(k-r+1)/r!}\right)^k\\
&\leq& \textstyle
(ne^{-p(k-r+1)/r!})^k\\
&=&n^{-\frac{1}{2}\left(\frac{3r!\ln n}{2p}+(r-1)\right)}
\end{eqnarray*}
so that
$\textnormal{Pr}(\alpha(G)<k)\geq 1-n^{-\frac{1}{2}\left(\frac{3r!\ln n}{2p}+(r-1)\right)}$,
as desired.
\end{proof}

This lemma enables a result for 3-uniform hypergraphs.
\begin{theorem}\label{thm:prob_bound}
For 3-uniform hypergraphs,
\[
\textnormal{Pr}\left(\lvert \UsizeC \rvert  \leq  O(\lvert C \rvert^2\ln \lvert C \rvert+\lvert C \rvert^2 x)+o(\lvert C \rvert^3) \right)
\geq 1-n^{-\frac{1}{2}\left(\frac{9\ln \lvert C \rvert}{p}+2\right)}
\]
with $x$ as defined in \cite{Damaschke-2009-unions}.
\end{theorem}
\begin{proof}
First, $\alpha(G)=|G|-k^*$~\cite[Theorem 3.15]{Halldrsson-2009-independent}, so
$\alpha(C)=\lvert C \rvert-k^*$ and
$k^*=\lvert C \rvert-\alpha(C)\geq \lvert C \rvert- 9\ln \lvert C \rvert / p -2$,
where the inequality follows from \cref{lem:independence}
and the inequality holds with probability at least $1-n^{-\frac{1}{2}\left(\frac{9\ln \lvert C \rvert}{p}+2\right)}$.
For 3-uniform hypergraphs, $\lvert U(k) \rvert\leq \frac{1}{4}k^*(k^2-(k^*)^2+2k^*x)+o(k^3)$~\cite[Theorem 18]{Damaschke-2009-unions}.
\end{proof}

\Cref{thm:prob_bound} enables us to find a small set $J$ that contains the core $C$ with high probability.
\begin{theorem}
  Let $C$ with $\lvert C \rvert=k$ be a planted hitting set in a 3-uniform hypergraph
  drawn from the SBM parameterized by $p$ and $q$ with $ck$ nodes for some $c \geq 1$.
  Then with high probability in $k$, there is a set $J$ with $C \subset J$ and
  $\lvert J \rvert\leq O(k^2\ln k+k^2x)+o(k^3)$.
\end{theorem}
\begin{proof}
The number of ways to link a vertex in the core to two nodes outside the core is $\binom{k(c-1)}{2}$
so that the probability of a node having at least one hyperedge to two nodes outside the core is $w = 1-(1-q)^{\binom{k(c-1)}{2}}$.
Thus $w^k$ is the probability of each of the nodes in the core having at least one hyperedge to two nodes outside the core.
This probability tends to 1 as $k \to \infty$, so by using \cref{lem:outside}, $C \subseteq \UsizeC$
with probability tending to 1. Setting $J = \UsizeC$ and using \cref{thm:prob_bound} gives the result.
\end{proof}

Of course, the above theorem is only useful if we have a bound on the quantity
$x$ since only in that case are we able to improve our bound on the size of $J$
from $O(k^3)$ to $O(k^2\ln k)$. Absent any additional information, $x\leq k^*$~\cite{Damaschke-2009-unions}.

In the following section, we build a practical algorithm for
recovery of planted hitting sets based on the union of minimal hitting sets, the theory for which was examined in this section.

%
%


\section{A practical recovery algorithm}\label{sec:UMHS}
Based on the theoretical results described above, we now develop a practical
algorithm for recovering a planted hitting set in a hypergraph. To put our
theory into practice, we place several of our theoretical results in
context. First, \cref{lem:overlap} says that \cref{alg:greedy} produces outputs
that \emph{must} overlap with the planted hitting set. Moreover, this was true
\emph{regardless of the order in which we processed the edges}. Thus, the basic
idea of our algorithm is simple: we find many hitting sets $S$ from
\cref{alg:greedy} and take their union.

By \cref{cor:plantedintersection}, if the planted hitting set $C$ is close to a
minimum hitting set, then any iteration of \cref{alg:greedy} will recover a large
fraction of it. Furthermore, by \cref{cor:overlap}, the outputs of
\cref{alg:greedy} must overlap by a modest amount, so taking the unions of
outputs cannot grow too fast. To limit growth further, we can prune the output
of \cref{alg:greedy} to be a minimal hitting set (this also tends to give
better results in practice). Now, \cref{thm:UMHS} says that the union
of pruned outputs will be bounded, provided the graph is large enough
and that the outputs are small enough. This turns out to be the case
in our experimental results. To summarize, our procedure is as follows:
(i) Find a hitting set $S$ from \cref{alg:greedy}, processing the hyperedges in a random order;
(ii) Prune $S$ to be a minimal hitting set $S'$; and
(iii) Repeat steps (i) and (ii) several times and output the union of the $S'$ found in step (ii).
Again, the key result is that \cref{lem:overlap} holds regardless of the
ordering in which the hyperedges are processed. The algorithm is formalized in
\cref{alg:UMHS}. A similar algorithm for the special case of $2$-uniform
hypergraphs (i.e., graphs) was recently analyzed~\cite{Benson-2018-planted}.

\begin{algorithm}[tb]
\caption{Union of Minimal Hitting Sets (UMHS) algorithm}
\label{alg:UMHS}
\begin{algorithmic}[1]
\State \textbf{Input}:
Hypergraph $G=(V,E)$ and number of iterations $N$\\
\textbf{Output}: 
Approximation $S'$ to the planted hitting set in $G$
\State $S'\gets \emptyset$
\For {$n\in \{1,....,N\}$}
\State $S\gets \mbox{Algorithm1(}G)$
\State Prune $S$ to be minimal hitting set of $G$
\State $S'\gets S'\cup S$
\EndFor
\State \textbf{return} $S'$
\end{algorithmic}
\end{algorithm}

The output of \cref{alg:UMHS} is by construction a
\emph{union of minimal hitting sets} (UMHS)
of the input hypergraph $G$, and we will refer to the algorithm as UMHS. The
theoretical guarantees from \cref{sec:minimal} apply here. Specifically, under
the assumption that the planted hitting set $C$ is itself minimal, we are
guaranteed full recovery with sufficient iterations. However, planted hitting
sets are rarely minimal in practice (in fact, they are not for the datasets we
consider). However, \cref{thm:UMHS} still guarantees that the output of the
algorithm will not grow very quickly.

Importantly, \cref{alg:UMHS} runs in time linear in the number of hyperedges.
In particular, suppose the input is a hypergraph $G=(V,E)$ of rank $r$ with a
minimum hitting set of size $k^*$. Then the algorithm produces an output in
$O(Nr^2k^*\lvert E \rvert)$ time. This happens because every call to
\cref{alg:greedy} takes $O(r\lvert E \rvert \rvert)$ time, while pruning $H$ to be a
minimal hitting set of $G$ requires $O(r^2k^* \lvert E \rvert)$ time in the
worst case (in practice, it is much faster). Thus, in the worst
case, pruning to a minimal hitting set takes more time than one iteration of
the greedy algorithm. However, this increase in processing time is not
prohibitive since it only increases by the constant factor $rk^*$. 
Nonetheless, pruning is crucial as in practice we find that it drastically reduces the output size. 


\section{Experimental results}

We now test \cref{alg:UMHS} on a number of real-world datasets and compare and
find that it consistently outperforms several baselines derived from network
centrality and core-periphery structure. We also show that we do not need too
many iterations within \cref{alg:UMHS} (the parameter $N$) for our performance
levels.

%


\begin{table*}[t]
\centering
\begin{minipage}[c]{0.19\textwidth}
\caption{Summary statistics of core-fringe hypergraph datasets.
  We construct $r$-uniform hypergraphs from six corpora for $r = 3,4,5$.
  The DBLP and tag datasets are collections of 50 hypergraphs, so we report the value range.
  }\label{tab:summary_stats}
\end{minipage}
\hfill
\begin{minipage}[c]{0.79\textwidth}
\setlength{\tabcolsep}{2pt}
\begin{tabular}{@{}rr ccc c ccc c ccc@{}}\toprule
  & & \multicolumn{3}{c}{number of nodes}
  & \phantom{!!!} & \multicolumn{3}{c}{size of core ($\lvert C \rvert$)}
  & \phantom{!!!} & \multicolumn{3}{c}{number of hyperedges}  \\
\cmidrule{3-5} \cmidrule{7-9} \cmidrule{11-13} 
& $r = $ & $3$ & $4$ & $5$ && $3$ & $4$ & $5$ && $3$ & $4$ & $5$ \\ \midrule
Enron & & 1,283 & 976 & 869 && 84 & 73 & 77 && 2,361 & 1,048 & 614 \\
Avocado & & 5,521 & 3,510 & 2,965 && 227 & 218 & 211 && 21,690 & 12,455 & 6,973 \\
W3C & & 1,778 & 749 & 308 && 353 & 237 & 125 && 1,882 & 389 & 88 \\
DBLP & & 3--875 & 4--514 & 5--261 && 3--114 & 4--90 & 5--55 && 1--753 & 1--265 & 1--78 \\
Math tags & & 93--1,180 & 97--1,239 & 121--1,153 && 3--15 & 4--15 & 5--15 && 112--13,771 & 128--17,424 & 98--13,410 \\
Ubuntu tags & & 96--2,005 & 402--2,043 & 225--1,991 && 3--15 & 4--15 & 5--15 && 94--12,990 & 437--13,210 & 220--9,905 \\
\bottomrule
\end{tabular}
\end{minipage}
\end{table*}


\subsection{Data}

The six datasets we use broadly fall into three broad classes based on the types
of planted hitting problems we derive from them. The first group is
three email datasets (Enron, Avocado, and W3C), where the planted hitting set is
a group of people at an organization, and hyperedges come from emails involving
multiple addresses. Next, in the DBLP dataset, a planted hitting set is a set
of authors who publish at the same conference in a given year, and we consider
several such conferences. Finally, in the Stack Exchange tagging datasets, the
core is a one-hop neighborhood of a node, which is a hitting set for the two-hop
neighborhood of the node; again, we consider several hitting sets. In this
sense, the DBLP and tagging datasets are collections of hypergraphs with
core-fringe structure.

\xhdr{Enron~\cite{Klimt-2004-Enron}, Avocado, and W3C~\cite{Craswell-2005-TREC}}
The hyperedges in these datasets are emails, where the
nodes are the sender and all receivers (more specifically, their email
addressses). Repeated hyperedges are discarded. In the Enron data, the core
is a set of 150 employees whose email was made public when the
company was under United States federal investigation. The Avocado corpus comes from a now-defunct technology
company, where the core are employee email accounts (we exclude email accounts associated
with conference rooms or mailing lists from the core).
Finally, the the W3C dataset comes from emails on mailing lists
related to W3C; the core is all email addresses with a
\texttt{w3c.org} domain.

\xhdr{DBLP}
DBLP is an online computer science bibliography. We construct a hypergraph from
coauthorship data in conference proceedings. We randomly sampled 50 conferences with relatively small cores
and constructed a set of core nodes $C$ from the authors of a given conference
in a randomly selected year. We then took the core-fringe hypergraph to be all
hyperedges involving at least one of these authors.

\xhdr{Math and Ubuntu tags~\cite{Benson-2018-simplicial}}
These datasets come from tags applied to questions on the Math and
AskUbuntu Stack Exchange web sites. Hyperedges are sets of tags that have been applied
to the same question. For both Stack Exchanges, we sampled 50 tags with relatively small cores uniformly at random. We then formed a core set $C$ from the one-hop neighborhood
of that tag, i.e., from the tag and all tags that appear as co-tags with the
tag. The set $C$ is then a planted hitting set for the two-hop neighborhood of
that tag.

We constructed 3-, 4-, and 5-uniform hypergraphs (derived as sub-hypergraphs)
from these datasets to facilitate comparison of UMHS performance to that of other centrality-based
algorithms (furthermore, Stack Exchange only has at most 5 tags in one
post). \Cref{tab:summary_stats} lists summary statistics of the datasets.


\begin{table*}[t]
  \centering
\setlength{\tabcolsep}{5pt}
\caption{%
%
%
  Planted hitting set recovery performance. We compare our proposed union of minimal
  hitting sets (UMHS) method against four hypergraph centrality measures
  for---degree~\cite{Kapoor-2013-centrality},
  clique graph eigenvector~\cite{Benson-2018-three},
  Z-eigenvector~\cite{Benson-2018-three},
  H-eigenvecgtor~\cite{Benson-2018-three}---as well as two core-periphery
  measurements--- Borgatti-Everett~\cite{Borgatti-2000-CP} and k-core~\cite{Seidman-1983-kcore}.
  Each method produces an ordering of vertices, and we measure performance
  by precision at the core size (fraction of top-$\lvert C \rvert$ ranked nodes that
  are in $C$) and area under the precision-recall curve for $r$-uniform hypergraphs
  ($r = 3,4,5$).
  The DBLP and tags datasets are collections of 50 hypergraphs, and we report the mean and
  standard deviation for these.
  UMHS scores outperforming all baselines by at least an 8\% relative improvement are bold.
  Any method's score outperforming UMHS by at least 8\% is also bold.
}
\label{tab:perf_stats}
\begin{tabular}{@{} r @{\quad\;\;} r c @{\quad} c @{\quad\;\;} cccc @{\quad} cc @{}}  \toprule
    & Dataset & $r$ & UMHS & Degree & Clique-eigen & Z-eigen & H-eigen & Borgatti-Everett & k-core \\ \midrule
Precision @ $\lvert C \rvert$
 & Enron       & 3 & \textbf{0.52}       & 0.33       & 0.14       & 0.11    & 0.13       & 0.14       & 0.33     \\
 &             & 4 & \textbf{0.53}       & 0.33       & 0.21       & 0.18    & 0.16       & 0.19       & 0.33       \\
 &             & 5 & \textbf{0.49}       & 0.25       & 0.16       & 0.14    & 0.17       & 0.14       & 0.17       \\ \cmidrule{2-10}
 & Avocado     & 3 &  \textbf{0.91}       & 0.60       & 0.56       & 0.54       & 0.58       & 0.59       & 0.58       \\
 &             & 4 & \textbf{0.82}       & 0.57       & 0.52       & 0.33       & 0.53       & 0.52       & 0.56       \\
 &             & 5 & \textbf{0.72}       & 0.56       & 0.49       & 0.34       & 0.50       & 0.49       & 0.33       \\ \cmidrule{2-10}
 & W3C         & 3 & \textbf{0.51}       & 0.38       & 0.20       & 0.11    & 0.21       & 0.19       & 0.35       \\
 &             & 4 & 0.43       & 0.43       & 0.30       & 0.30    & 0.31       & 0.30       & 0.40       \\
 &             & 5 & 0.30       & \textbf{0.42}       & 0.31       & \textbf{0.34}       & \textbf{0.33}       & \textbf{0.33}       & \textbf{0.41}       \\ \cmidrule{2-10}
 & DBLP        & 3 & 0.59$\pm$0.15 & \textbf{0.63}$\pm$0.15 & 0.44$\pm$0.24 & 0.42$\pm$0.28 & 0.41$\pm$0.26 & 0.41$\pm$0.26 & 0.44$\pm$0.21 \\
 &             & 4 & 0.48$\pm$0.15 & \textbf{0.58}$\pm$0.17 & 0.40$\pm$0.27 & 0.36$\pm$0.29 & 0.39$\pm$0.28 & 0.38$\pm$0.31 & 0.40$\pm$0.22 \\
 &             & 5 & 0.39$\pm$0.15 & \textbf{0.65}$\pm$0.22 & \textbf{0.56}$\pm$0.31 & \textbf{0.59}$\pm$0.35 & \textbf{0.57}$\pm$0.31 & \textbf{0.56}$\pm$0.31 & \textbf{0.56}$\pm$0.25 \\  \cmidrule{2-10}
 & Math tags   & 3 & \textbf{0.80}$\pm$0.14 & 0.58$\pm$0.13 & 0.43$\pm$0.15 & 0.41$\pm$0.14 & 0.45$\pm$0.13 & 0.50$\pm$0.14 & 0.43$\pm$0.12 \\
 &             & 4 & \textbf{0.71}$\pm$0.18 & 0.48$\pm$0.12 & 0.38$\pm$0.12 & 0.31$\pm$0.14 & 0.39$\pm$0.11 & 0.31$\pm$0.10 & 0.38$\pm$0.11 \\
 &             & 5 & \textbf{0.61}$\pm$0.19 & 0.40$\pm$0.10 & 0.33$\pm$0.10 & 0.24$\pm$0.11 & 0.32$\pm$0.10 & 0.31$\pm$0.10 & 0.33$\pm$0.10 \\  \cmidrule{2-10}
 & Ubuntu tags & 3 & \textbf{0.80}$\pm$0.09 & 0.58$\pm$0.16 & 0.36$\pm$0.13 & 0.32$\pm$0.13 & 0.40$\pm$0.15 & 0.55$\pm$0.16 & 0.36$\pm$0.12 \\
 &             & 4 & \textbf{0.78}$\pm$0.13 & 0.50$\pm$0.16 & 0.38$\pm$0.17 & 0.30$\pm$0.15 & 0.40$\pm$0.17 & 0.33$\pm$0.17 & 0.38$\pm$0.15 \\
 &             & 5 & \textbf{0.69}$\pm$0.18 & 0.38$\pm$0.14 & 0.28$\pm$0.12 & 0.23$\pm$0.13 & 0.27$\pm$0.12 & 0.25$\pm$0.12 & 0.28$\pm$0.14 \\  \midrule 
Area under PR curve
 & Enron       & 3 & \textbf{0.18}       & 0.15       & 0.08       & 0.07       & 0.07       & 0.08       & 0.15       \\
 &             & 4 & 0.16       & 0.16       & 0.10       & 0.09       & 0.09       & 0.10       & 0.16       \\
 &             & 5 & \textbf{0.17}       & 0.13       & 0.10       & 0.10       & 0.10       & 0.10       & 0.10       \\ \cmidrule{2-10}
 & Avocado     & 3 & \textbf{0.43}       & 0.38       & 0.33       & 0.31       & 0.36       & 0.37       & 0.36       \\
 &             & 4 & \textbf{0.38}       & 0.35       & 0.30       & 0.15       & 0.31       & 0.30       & 0.34       \\
 &             & 5 & 0.34       & 0.34       & 0.28       & 0.16       & 0.28       & 0.27       & 0.33       \\ \cmidrule{2-10}
 & W3C         & 3 & \textbf{0.36}       & 0.27       & 0.20       & 0.19       & 0.20       & 0.20       & 0.25       \\
 &             & 4 & \textbf{0.40}       & 0.37       & 0.31       & 0.31       & 0.31       & 0.31       & 0.35       \\
 &             & 5 & 0.41       & 0.41       & 0.38       & 0.38       & 0.38       & 0.38       & 0.41       \\ \cmidrule{2-10}
 & DBLP        & 3 & \textbf{0.57}$\pm$0.17 & 0.52$\pm$0.18 & 0.39$\pm$0.24 & 0.40$\pm$0.28 & 0.38$\pm$0.25 & 0.37$\pm$0.25 & 0.47$\pm$0.23 \\
 &             & 4 & 0.47$\pm$0.19 & 0.46$\pm$0.17 & 0.36$\pm$0.26 & 0.35$\pm$0.27 & 0.36$\pm$0.26 & 0.36$\pm$0.29 & 0.39$\pm$0.22 \\
 &             & 5 & 0.54$\pm$0.23 & 0.57$\pm$0.24 & 0.53$\pm$0.31 & \textbf{0.58}$\pm$0.34 & 0.54$\pm$0.31 & 0.54$\pm$0.31 & \textbf{0.58}$\pm$0.26 \\  \cmidrule{2-10}
 & Math tags   & 3 & \textbf{0.68}$\pm$0.18 & 0.36$\pm$0.14 & 0.21$\pm$0.13 & 0.19$\pm$0.12 & 0.23$\pm$0.12 & 0.27$\pm$0.14 & 0.26$\pm$0.12 \\
 &             & 4 & \textbf{0.56}$\pm$0.20 & 0.25$\pm$0.11 & 0.17$\pm$0.11 & 0.12$\pm$0.11 & 0.17$\pm$0.09 & 0.11$\pm$0.07 & 0.17$\pm$0.10 \\
 &             & 5 & \textbf{0.48}$\pm$0.21 & 0.18$\pm$0.08 & 0.13$\pm$0.07 & 0.08$\pm$0.07 & 0.12$\pm$0.06 & 0.11$\pm$0.06 & 0.14$\pm$0.07 \\ \cmidrule{2-10}
 & Ubuntu tags & 3 & \textbf{0.65}$\pm$0.14 & 0.37$\pm$0.17 & 0.15$\pm$0.11 & 0.13$\pm$0.11 & 0.19$\pm$0.13 & 0.33$\pm$0.17 & 0.20$\pm$0.12 \\
 &             & 4 & \textbf{0.64}$\pm$0.16 & 0.28$\pm$0.16 & 0.18$\pm$0.14 & 0.12$\pm$0.12 & 0.19$\pm$0.15 & 0.14$\pm$0.14 & 0.18$\pm$0.13 \\
 &             & 5 & \textbf{0.54}$\pm$0.18 & 0.17$\pm$0.11 & 0.10$\pm$0.08 & 0.07$\pm$0.09 & 0.09$\pm$0.08 & 0.08$\pm$0.08 & 0.12$\pm$0.10 \\
\bottomrule
\end{tabular}
\end{table*}

\subsection{Recovery results}

We tested the UMHS algorithm on several datasets, and it outperforms other
algorithms consistently.
For baselines, we use two techniques. First, we use notions of network
centrality developed for hypergraphs with the idea that nodes in the core could
be identified via large centrality scores. Specifically, we compare against
(i) hypergraph degree centrality (the
number of hyperedges in which a node appears)~\cite{Kapoor-2013-centrality};
(ii) clique graph eigenvector (eigenvector centrality
on the weighted clique graph, where $w_{ij}$ is the number of hyperedges
containing $i$ and $j$)~\cite{Benson-2018-three};
(iii) Z-eigenvector centrality (based on
$Z$-eigenvectors of tensors)~\cite{Benson-2018-three}; and
(iv) H-eigenvector centrality (based on
$H$-eigenvectors of tensors)~\cite{Benson-2018-three}.
Second, we use notions of network core-periphery decompositions, where we expect
that nodes in $C$ will be identified as ``core'' in this sense. We use two
algorithms:
(i) the $k$-core decomposition~\cite{Seidman-1983-kcore} based on hypergraph
degree; and
(ii) Borgatti-Everett scores in the weighted clique graph.

All of these methods induce an ordering on the nodes. We induce an ordering on
the output of \cref{alg:UMHS} by degree and then order the remaining nodes
(those not in the output of \cref{alg:UMHS}) in order by degree.
With an ordering
on the nodes, we measure performance in terms of precision at core size,
i.e., the fraction of the first $\lvert C \rvert$ in the ordering that are in $C$,
as well as area under the precision-recall curve (AUPRC). We use AUPRC as opposed
to area under the ROC curve due to class imbalance~\cite{Davis-2006-curves}, namely, most
nodes are not in the core.

\Cref{tab:perf_stats} reports the results of all methods on all datasets.
In terms of precision at core size, UMHS out-performs the baselines by
wide margins on the Enron, Avocado, Math tags, and Ubuntu tags datasets
for all uniformities $r$. UMHS does well for the 3-regular W3C hypergraphs.
On DBLP and the other W3C hypergraphs, the simple degree heuristic
seems to perform well, although our algorithm still outperforms it on a large share of samples.
We see similar trends when measuring in terms of AUPRC, with UMHS dominating on the same set of datasets above; and in this case, UMHS is more competitive even for the datasets where it is weakest --- the 4-uniform and 5-uniform W3C
hypergraphs, as well as the DBLP hypergraphs.

Overall, the performance of many algorithms degrades as we increase the uniformity
of the hypergraph.  In general, this makes sense---the core nodes are increasingly
hidden in larger hyperedges. In the case of UMHS, we have a more specific interpretation of what this means.
As the hyperedges get larger, there may be cases of hyperedges
that contain just one node in the core $C$, and the rest in the fringe. However,
by the greedy structure of \cref{alg:greedy}, we would put all nodes in the hitting set.
Furthermore, the bounds in Corollaries \ref{cor:plantedintersection} and \ref{cor:overlap} also degrade
as we increase the uniformity $r$. Nevertheless, our proposed UMHS algorithm
still outperforms the baselines at an aggregate level.

\subsection{Recovery as a function of output size}

Our UMHS algorithm (\cref{alg:UMHS}) has a single tuning parameter, which is the
number of iterations $N$, i.e., the number of calls to the sub-routine for greedy
maximal matching (\cref{alg:greedy}).  Here we examine performance of UMHS
as a function of $N$. Specifically, we analyze (i)
the fraction of core nodes in the planted hitting set that are recovered and
(ii) the output size of \cref{alg:UMHS} as a function of the number of
iterations (\cref{fig:size} shows results for the email datasets).

We highlight a couple of important findings. First, we only need around 50
iterations to achieve high recovery rates. Each iteration is fast, and the
entirety of the algorithm's running time takes at most a few minutes on the
larger datasets. Second, the union of minimal hitting sets size tends to
increase sharply with a few iterations and then levels off sharply. These results are
consistent with our theory. \Cref{thm:UMHS,cor:plantedintersection} both provide
theoretical justification for why the output should not grow too large.

\begin{figure}[tb]
\begin{center}
    \includegraphics[width=0.49\columnwidth]{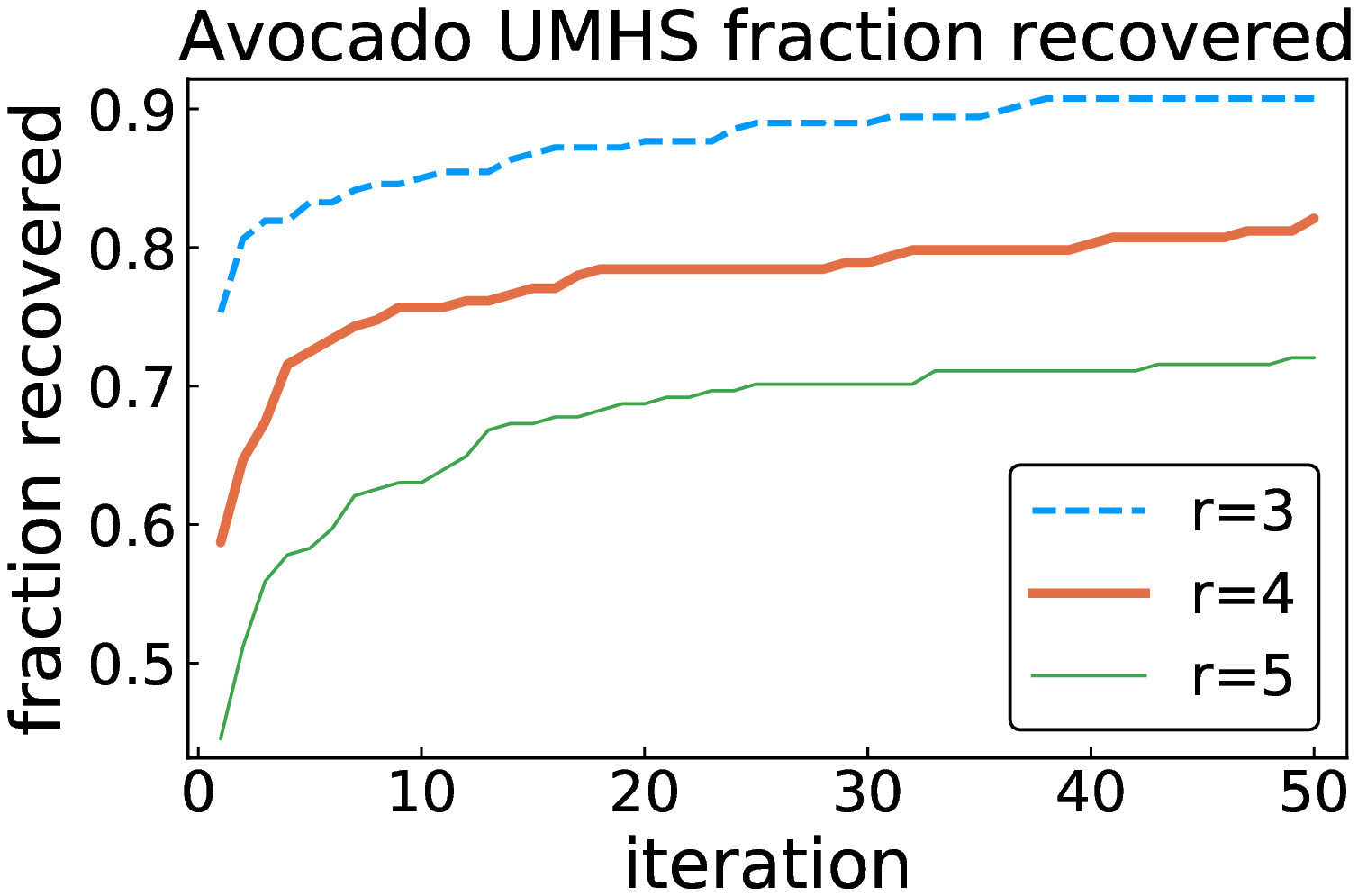}
    \includegraphics[width=0.49\columnwidth]{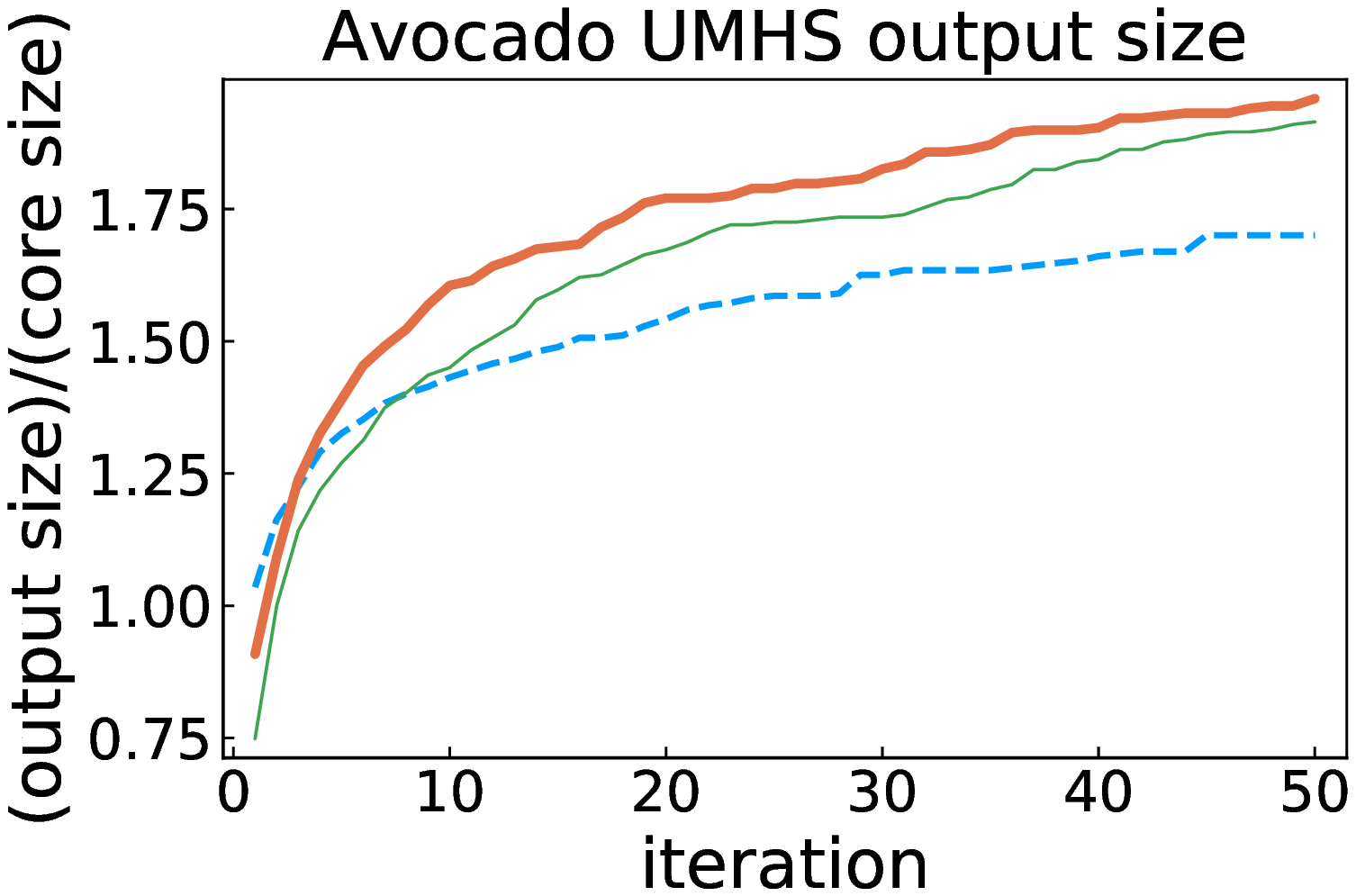} \\
    \includegraphics[width=0.49\columnwidth]{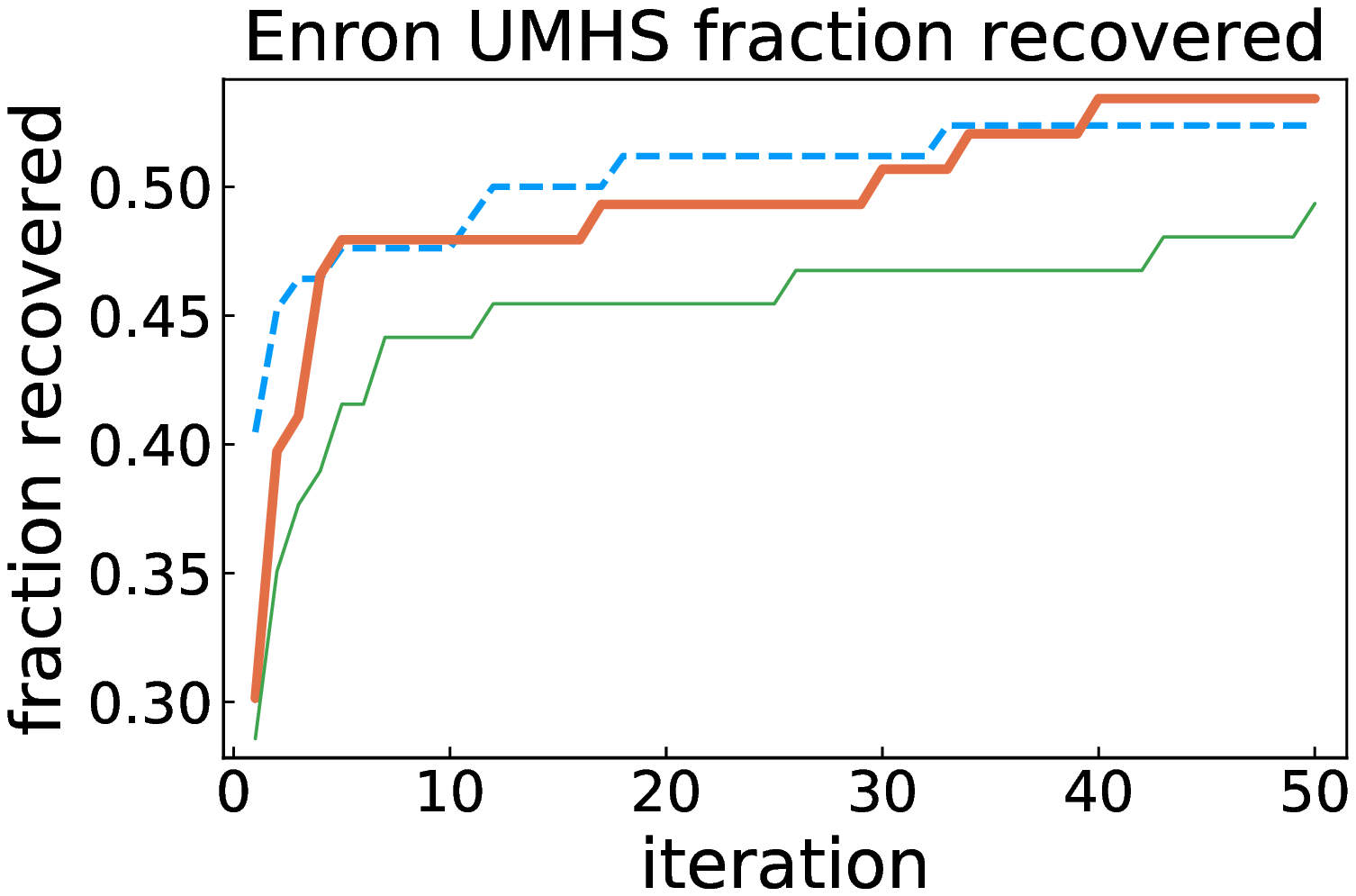} 
    \includegraphics[width=0.49\columnwidth]{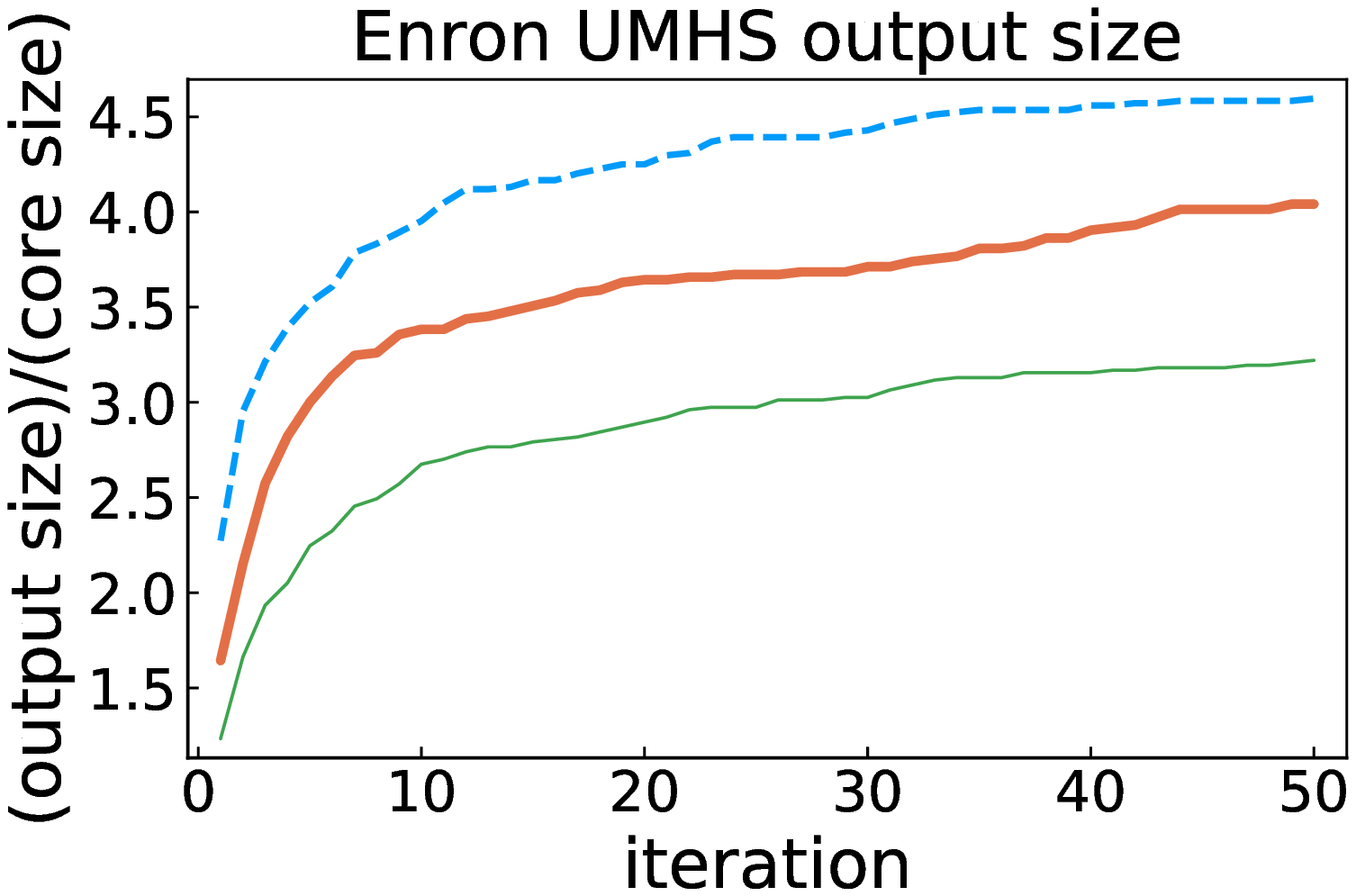} \\
    \includegraphics[width=0.49\columnwidth]{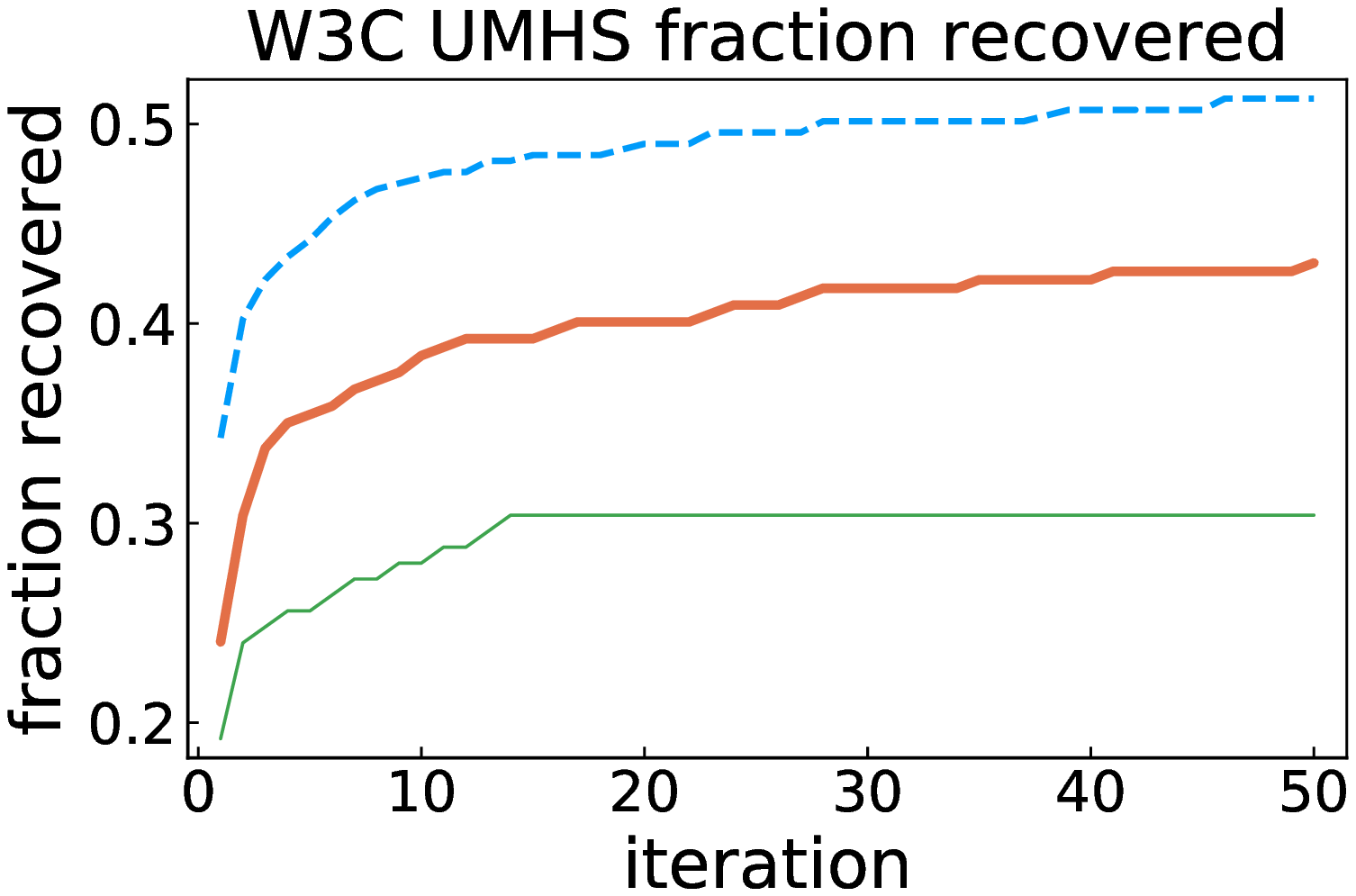}
    \includegraphics[width=0.49\columnwidth]{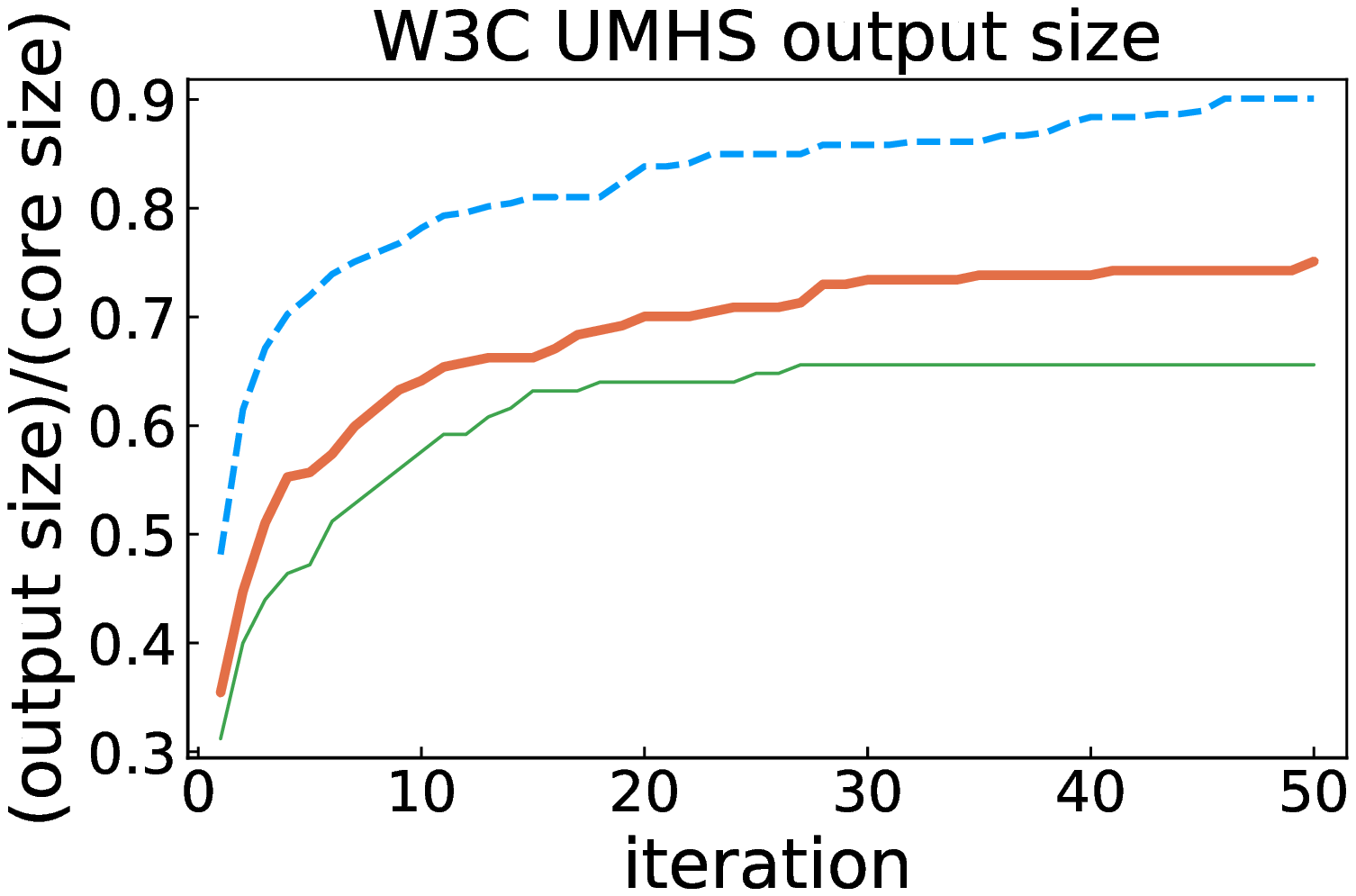}\\
\end{center}
\caption{%
Planted hitting set recovery improves with more iterations before leveling
(left column), and similarly, size increases with more iterations before leveling
(right column). Fewer than 50 iterations typically reaches peak performance in
our email datasets. The leveling of the output size is consistent with
\cref{thm:UMHS,cor:plantedintersection}, which say that our algorithm's output
can't grow too large.
}
\label{fig:size}
\end{figure}


\section{Related work}

On the theoretical side, our problem can be thought of as an
instance of a ``planted'' problem, where a certain
type of graph structure is planted or hidden in an a graph and
one must recover the latent structure given the graph. 
Well-studied problems in this space include the
planted clique, where one tries to find a clique placed  
in a sample from a $G_{n,1/2}$ graph~\cite{Alon-1998-cliques,Feige-2000-finding,Dekel-2011-hidden,Deshpande-2014};
and the planted partition or stochastic block model recovery, where a random graph is sampled
with probabilities dependent on latent labels of nodes, and the goal is to recover
these labels~\cite{Mossel-2014-belief,Tsourakakis2015,Peixoto-2017-nonparametric,Abbe-2018-community}.
These planted problems are based on some random way in which the graph was sampled.
In our case of planted hitting sets, the graph was deterministic, although we
could improve our results under a random hypergraph model.
Most related to our results is recent work in planted vertex covers; this is 
a special case of hitting sets for the case of graphs (which mathematically are the same as 2-uniform hypergraphs)~\cite{Benson-2018-planted}. As discussed above, the hypergraph model is
more realistic for many datasets (especially email), given its ability to represent groups of more than two individuals at a time.

Within the field of network science, the idea of a small planted hitting set 
fits with two related ideas: node centrality and core-periphery structure.
The former concept deals with finding important nodes (or ranking them) 
based on their connections, often provided as a graph~\cite{Bonacich-1987-power,Boldi-2014-axioms,Gleich-2015-PageRank}.
Nodes in hitting sets are central to hypergraphs almost by definition---every
hyperedge must contain at least one of these nodes. Thus, we expect them
to be ``central'' in some sense. However, we found that existing measures
of node centrality in hypergraphs did not recover planted hitting sets
at the same levels as our union of minimal hitting sets algorithm.

Core-periphery structure is a mesoscale property of many networks, where
there is a densely connected core set of nodes along with a loosely
connected periphery~\cite{Csermely-2013-CP}. 
Such a composition has been studied in sociology~\cite{Lorrain-1971-equivalence,Doreian-1985-equivalence} and 
international trade~\cite{Smith-1992-structure}, 
where the core-periphery structure is due to differential status. 
Now, core-periphery identification is a broader tool for identifying
structure in general networks~\cite{Holme-2005-CP,Rombach-2017-CP,Govindan-2017-kpeak,Sarkar-2018-richclubs}. The planted hitting set that we aim to recover corresponds to an extreme type of core-periphery structure;
due to the way in which we assume the hypergraph is measured, nodes 
on the periphery (the ``fringe nodes'') cannot be connected without
a core node as an intermediary in a hyperedge.

Finally, core-fringe structure itself has received some attention.
Romero et al.\ analyzed the behavior of a core group of employees
at a hedge fund in the context of their relationships with contacts outside
of the company~\cite{Romero-2016-stress}, and Benson and Kleinberg
analyzed how links between core and fringe nodes influence graph-based link prediction
algorithms~\cite{Benson-2018-core}. Our research highlights additional
richness to the problem when the underlying data model is a hypergraph.


\section{Discussion}

Network data is a partial view of a larger
system~\cite{Laumann-1989-boundary}. A common case is when the interactions of
some specified set of actors or nodes are under surveillance. This provides a
``core-fringe'' structure to the network---we can see all interactions involving
the core but only the interactions of the fringe with the core. When data is
leaked or metadata is lost over time due to data provenance issues, we would
like to be able to recover these core and fringe labels for security or data
maintenance purposes.

Here, we have studied this problem where the network data is a
hypergraph, so the core is a hitting set. This setting is common in email data or
situations in which groups of people are meeting. We used co-authorship as a
proxy for the latter situation, but one can imagine situations in which one
records the groups of attendees of meetings involving someone under
surveillance.
Theoretically, we showed that the union of minimal hitting sets cannot be too
large and that the output of the well-known approximation algorithm for minimum
hitting sets has to somehow overlap the core.

Using these results as
motivation, we developed an extremely simple algorithm for recovering the core:
take the union of minimal hitting sets that are output by randomly
initialized instances of the approximation algorithm. This method out-performed
several strong baselines, including methods based on centrality and
core-periphery structure.

However, our simple algorithm opens several
avenues for improvement in future work. For instance, our model assumed an
undirected and unweighted hypergraph structure. There are models of directed and
weighted hypergraphs~\cite{Gallo-1993-directed-hgraph} that could be used to
improve the recovery algorithm. In addition, theory on the number of calls to
the approximation algorithm subroutine would be useful. In practice, only a few
calls is sufficient and perhaps assuming particular structure on the hypergraph
could yield additional theoretical insight.

\bibliographystyle{ACM-Reference-Format}

\appendix

\clearpage

\end{document}